\documentclass[11pt]{article}
\usepackage[T1]{fontenc}
\usepackage[utf8]{inputenc}
\usepackage{lmodern}
\usepackage{microtype}
\usepackage[letterpaper,margin=1in]{geometry}
\usepackage[UKenglish]{babel}
\usepackage{url}
\usepackage[numbers,sort&compress]{natbib}
\usepackage{color}
\usepackage{doi}
\usepackage{latexsym}
\usepackage{mathrsfs}
\usepackage{amsmath}
\usepackage{amssymb}
\usepackage{amsthm}
\usepackage[all,warning]{onlyamsmath}
\usepackage{enumitem}
\usepackage{xparse}
\usepackage{mathtools}
\usepackage{graphicx}
\usepackage{complexity}
\usepackage[small]{caption}
\usepackage{hyperref}

\theoremstyle{plain}
\newtheorem{theorem}{Theorem}
\newtheorem{lemma}[theorem]{Lemma}

\newtheorem{corollary}[theorem]{Corollary}

\theoremstyle{definition}
\newtheorem{definition}[theorem]{Definition}

\newtheorem{problem}[theorem]{Problem}

\theoremstyle{remark}

\setlist[itemize]{label=--}
\setlist[enumerate]{label=(\arabic*),labelindent=\parindent,leftmargin=*}

\DeclarePairedDelimiter\braces{\{}{\}}
\NewDocumentCommand\set{O{}mg}{\ensuremath{\braces[#1]{#2\IfNoValueTF{#3}{}{\,:\,#3}}}}

\DeclareMathOperator{\indeg}{in-deg}
\DeclareMathOperator{\outdeg}{out-deg}
\DeclareMathOperator{\dist}{dist}
\DeclareMathOperator{\vbl}{vbl}
\DeclareMathOperator{\mypoly}{poly}

\newcommand{\N}{\mathbb{N}}
\newcommand{\Real}{\mathbb{R}}
\newcommand{\lllemma}{Lov\'asz local lemma}
\newcommand{\LLLemma}{Lov\'asz Local Lemma}
\newclass{\local}{LOCAL}
\newcommand{\Ep}{\mathcal{E}}

\newcommand{\namedref}[2]{\hyperref[#2]{#1~\ref*{#2}}}
\newcommand{\sectionref}[1]{\namedref{Section}{#1}}

\newcommand{\theoremref}[1]{\namedref{Theorem}{#1}}

\newcommand{\figureref}[1]{\namedref{Figure}{#1}}
\newcommand{\lemmaref}[1]{\namedref{Lemma}{#1}}

\newcommand{\defref}[1]{\namedref{Definition}{#1}}

\newenvironment{myabstract}
               {\list{}{\listparindent 1.5em%
                        \itemindent    \listparindent
                        \leftmargin    1cm
                        \rightmargin   1cm
                        \parsep        0pt}%
                \item\relax}
               {\endlist}

\newenvironment{mycover}
               {\list{}{\listparindent 0pt
                        \itemindent    \listparindent
                        \leftmargin    1cm
                        \rightmargin   1cm
                        \parsep        0pt}%
                \raggedright
                \item\relax}
               {\endlist}

\newcommand{\myemail}[1]{\,$\cdot$\, {\small #1}\par\vspace{2pt}}
\newcommand{\myaff}[1]{{\small #1\par}\bigskip}

\hypersetup{
  colorlinks=true,
  linkcolor=black,
  citecolor=black,
  filecolor=black,
  urlcolor=[rgb]{0,0.1,0.5},
  pdftitle={A Lower Bound for the Distributed Lov\'asz Local Lemma},
  pdfauthor={Sebastian Brandt, Orr Fischer, Juho Hirvonen, Barbara Keller, Tuomo Lempiäinen, Joel Rybicki, Jukka Suomela, Jara Uitto}
}

\begin{document}

\mbox{}
\begin{mycover}
{\huge \bfseries A Lower Bound for the Distributed \\ Lov\'asz Local Lemma \par}
\bigskip
\bigskip

\textbf{Sebastian Brandt}
\myemail{brandts@tik.ee.ethz.ch}
\myaff{ETH Zurich, Switzerland}

\textbf{Orr Fischer}
\myemail{orrfischer@mail.tau.ac.il}
\myaff{School of Computer Science, Tel Aviv University, Israel}

\textbf{Juho Hirvonen}
\myemail{juho.hirvonen@aalto.fi}
\myaff{Helsinki Institute for Information Technology HIIT,\\
Department of Computer Science, Aalto University, Finland}

\textbf{Barbara Keller}
\myemail{barkelle@tik.ee.ethz.ch}
\myaff{ETH Zurich, Switzerland}

\textbf{Tuomo Lempiäinen}
\myemail{tuomo.lempiainen@aalto.fi}
\myaff{Helsinki Institute for Information Technology HIIT,\\
Department of Computer Science, Aalto University, Finland}

\textbf{Joel Rybicki}
\myemail{joel.rybicki@aalto.fi}
\myaff{Helsinki Institute for Information Technology HIIT,\\
Department of Computer Science, Aalto University, Finland}

\textbf{Jukka Suomela}
\myemail{jukka.suomela@aalto.fi}
\myaff{Helsinki Institute for Information Technology HIIT,\\
Department of Computer Science, Aalto University, Finland}

\textbf{Jara Uitto}
\myemail{uitto@bitsplitters.com}
\myaff{Bitsplitters GmbH, Switzerland}
\end{mycover}

\bigskip
\begin{myabstract}
\noindent\textbf{Abstract.}
We show that any randomised Monte Carlo distributed algorithm for the Lov\'asz local lemma requires $\Omega(\log \log n)$ communication rounds, assuming that it finds a correct assignment with high probability. Our result holds even in the special case of $d = O(1)$, where $d$ is the maximum degree of the dependency graph. By prior work, there are distributed algorithms for the Lov\'asz local lemma with a running time of $O(\log n)$ rounds in bounded-degree graphs, and the best lower bound before our work was $\Omega(\log^* n)$ rounds [Chung et al.\ 2014].
\end{myabstract}
\thispagestyle{empty}
\setcounter{page}{0}
\newpage

\section{Introduction}\label{sec:intro}

In this work, we give a lower bound for the constructive \lllemma{} (LLL) in the context of distributed algorithms. We study the running time as a function of $n$ (the number of events), and prove a lower bound that holds even if $d$ (the maximum degree of the dependency graph) is bounded by a constant. By prior work, there are distributed algorithms for LLL with a running time of $O(\log n)$ communication rounds in this case and $o(\log n)$ rounds for restricted variants~\cite{chung14distributed}, and it is known that any distributed algorithm for LLL requires $\Omega(\log^* n)$ rounds~\cite{chung14distributed,linial92locality,naor91lower}. We prove a new lower bound of $\Omega(\log \log n)$ rounds.

\subsection{Distributed \LLLemma \label{ssec:dist-lll}}

Recall the following symmetric version of LLL:

\begin{theorem}[\lllemma]\label{thm:lll}
  Let $\mathcal{E} = \{ E_1, \ldots, E_n \}$ be a finite set of events such that each $E_i$ depends on at most $d$ other events. If $\Pr(E_i) \le p$ and $ep(d+1) \le 1$, then there is a positive probability that none of the events occur.
\end{theorem}

We consider \emph{distributed} algorithmic variants of LLL. The basic framework is as follows. Let $\mathcal{X} = \{X_1, \ldots, X_m\}$ be a set of mutually independent random variables and assume that each $E_i$ depends only on variables in $\mathcal{X}$; denote by $\vbl(E_i) \subseteq \mathcal{X}$ the subset of variables that event $E_i$ depends on. Form the \emph{dependency graph} $G_\mathcal{E} = (\mathcal{E}, \mathcal{D})$, where $\mathcal{D} = \set{\set{E_i,E_j}}{\vbl(E_i) \cap \vbl(E_j) \neq \emptyset}$. Now consider a distributed system in which the communication network is identical to the graph~$G_\mathcal{E}$: each node of the system is associated with a bad event $E \in \mathcal{E}$, and two nodes are adjacent if and only if their associated events depend on at least one common variable. The task is for each node to find an assignment to its variables $\vbl(E)$ such that adjacent nodes agree on the values of their common variables and all the events in $\mathcal{E}$ are avoided.

We use the standard $\local$ model of distributed computing~\cite{linial92locality,peleg00distributed}. Initially each node is only aware of its own part of the input, but the nodes can exchange messages to learn more about the structure of the problem instance. Eventually, each node has to stop and output its own part of the variable assignment. Communication takes place in synchronous communication rounds, and the running time is defined to be equal to the number of communication rounds. 
Following the common practice, we say that an event occurs with high probability if it occurs with probability at least $1 - 1/n^c$, where $c$ is an arbitrarily large constant.
We consider randomised Monte Carlo algorithms, where all bad events are avoided with high probability and the running time is deterministic (more precisely, some function of $n$).

\subsection{Main Result and Key Techniques}

We prove the following lower bound for LLL algorithms in the $\local$ model: any randomised Monte Carlo  algorithm that produces a correct solution with high probability requires $\Omega(\log \log n)$ communication rounds, even if we restrict our input instances to $d$-regular graphs with $d = O(1)$. This is a substantial improvement over the lower bound of $\Omega(\log^* n)$ from prior work~\cite{chung14distributed,linial92locality,naor91lower}.

To derive the lower bound, we introduce two new graph problems that are closely related to each other: \emph{sinkless orientation} and \emph{sinkless colouring} (see \sectionref{sec:prelim}). Then we proceed as follows:
\begin{enumerate}
    \item We show that any Monte Carlo distributed algorithm for LLL implies a Monte Carlo distributed algorithm for sinkless orientation in $3$-regular graphs, with asymptotically the same running time (see \sectionref{sec:reduction}).
    \item We show that any Monte Carlo distributed algorithm for sinkless orientation in $3$-regular graphs has a running time of $\Omega(\log \log n)$.
\end{enumerate}
For the second step, we study the sinkless orientation problem in high-girth graphs. The key ingredient is a \emph{mutual speedup lemma} (see \sectionref{sec:speedup}) that holds in graphs of girth larger than $2t+1$:
\begin{enumerate}[noitemsep]
    \item If we can find a sinkless colouring in $t$ rounds, we can find a sinkless orientation in $t$ rounds.
    \item If we can find a sinkless orientation in $t$ rounds, we can find a sinkless colouring in $t-1$ rounds.
\end{enumerate}
By iterating the mutual speedup lemma, we can then obtain an algorithm for finding a sinkless orientation in high-girth graphs with a running time of $0$ rounds, which is absurd. The mutual speedup lemma amplifies the failure probability, but not too much---if the original algorithm works with high probability, we can still reach the contradiction after $o(\log \log n)$ iterations.

As a by-product, we also obtain a lower bound for $d$-colouring $d$-regular high-girth graphs: any proper node colouring with $d$ colours is also a sinkless colouring (while the converse is not true).

Our lower-bound proof does not make use of the full power of LLL---it also holds if we replace the usual assumption of $ep(d+1) \le 1$ with, e.g., the classical formulation~\cite{erdos75local} of Erd\H{o}s and Lov\'asz where the assumption is $4pd \leq 1$. In addition, our bound holds for the symmetric version of LLL, and therefore, it trivially applies to the asymmetric LLL~\cite{alon08probabilistic} as well.

\subsection{Related Work}

The celebrated \lllemma{} was first introduced in 1975~\cite{erdos75local} and has since then found applications in proving the existence of various combinatorial structures~\cite{alon08probabilistic,mitzenmacher05probability,molloy02colouring}. However, the original proof was non-constructive, and thus, did not yield an efficient (centralised) algorithm for \emph{finding} such a structure.

Beck~\cite{beck91algorithmic} showed that constructive versions of the local lemma do exist, albeit with weaker guarantees: there exists a deterministic algorithm that finds a satisfying assignment to a certain variant of LLL in polynomial time. This breakthrough result stimulated a long line of research in devising new \emph{algorithmic} versions of the local lemma with more general conditions and better performance~\cite{alon91parallel,molloy98further,czumaj00algorithmic,srinivasan08improved,moser09constructive,moser10constructive,haeupler10constructive,chandrasekaran13deterministic}. The algorithmic LLL has found numerous applications e.g.\ in the context of colouring, scheduling, and satisfiability problems~\cite{molloy02colouring,czumaj00algorithmic,elkin15matching,chung14distributed,moser09constructive,czumaj00coloring,leighton99fast,srinivasan06extension}.

A key breakthrough was the result by Moser and Tardos~\cite{moser10constructive}: they showed that even a very general form of the local lemma has a constructive counterpart; a natural resampling algorithm finds a satisfying assignment efficiently. Moser and Tardos also gave a parallel variant of this algorithm which can easily be implemented in a distributed setting as well. Indeed, already Alon~\cite{alon91parallel} observed that LLL admits parallelism by showing how to parallelise Beck's original approach~\cite{beck91algorithmic}. Subsequently, many papers have also considered how to attain efficient parallel and distributed algorithms for LLL~\cite{moser10constructive,chandrasekaran13deterministic,chung14distributed,haeupler15improved}.

The algorithmic framework of Moser and Tardos~\cite{moser10constructive} is based on an iterative random sampling method. The idea is to start with a random assignment and while a violated constraint exists, the algorithm then iteratively resamples variables in some violated constraint. Resampling is continued until no more violated constraints exist. The algorithm is easy to parallelise by noting that one can resample variables in \emph{independent} constraints, that is, in constraints that do not share variables. Now it suffices to pick an independent set in the subgraph of the dependency graph induced by the violated constraints and resample variables related to these constraints. Moser and Tardos use Luby's algorithm~\cite{luby86simple} to find a maximal independent set in $O(\log n)$ rounds in each resampling iteration. In total, this algorithm requires $O(\log n)$ resampling iterations thus leading to a total running time of $O(\log^2 n)$ rounds in the distributed setting.

One approach for speeding up this basic algorithm is to use faster algorithms for computing the independent sets. For example, in constant-degree graphs, a maximal independent set can be found in $\Theta(\log^* n)$ rounds~\cite{linial92locality}. More generally in low-degree graphs, maximal independent sets can be found in $O(d + \log^* n)$ rounds~\cite{barenboim14distributed} and $O(\log d \cdot \sqrt{\log n})$ rounds~\cite{barenboim12locality}, thus making it possible to solve LLL in $o(\log^2n)$ rounds when the degrees are small. However, this approach has an inherent barrier: there exist graphs, where any algorithm that finds a maximal independent set needs $\Omega(\min\{ \log d, \sqrt{\log n}\})$ rounds~\cite{kuhn10local}. 

As pointed out by Moser and Tardos~\cite{moser10constructive}, it is not necessary to find a \emph{maximal} independent set, but a large independent set suffices. Following this idea, Chung et al.~\cite{chung14distributed} gave a distributed algorithm where they instead compute so-called \emph{weakly maximal} independent sets, where the probability that a node is not in the produced independent set $S$ or neighbouring a node in set $S$ is bounded by $1/\mypoly(d)$. They showed that this can be done in $O(\log^2 d)$ rounds, thus the dependency on $n$ in the total running time of the LLL algorithm is only $O(\log n)$. Recently, Ghaffari improved this further by showing that weakly maximal independent sets can be computed in $O(\log d)$ rounds~\cite{ghaffari15mis}. 

If one is interested in weaker forms of LLL, Chung et al.~\cite{chung14distributed} also provide faster algorithms running in $O(\log n / \log \log n)$ rounds for the LLL criteria $pf(d) < 1$, where $f(d)$ is an exponential function.

While there are numerous positive results, only a few lower bounds for LLL are known. Moser and Tardos point out that in their resampling approach, $\Omega(\log_{1/p} n)$ expected iterations of resampling are needed. Recently, Haeupler and Harris \cite{haeupler15improved} conjectured that parallel resampling algorithms need $\Omega(\log^2 n)$ time.

In the distributed setting, Chung et al.~\cite{chung14distributed} show an unconditional lower bound showing that essentially any distributed LLL algorithm takes $\Omega(\log^* n)$ rounds. This bound follows from the fact that LLL can be used to properly colour a ring using only a constant number of colours, which is known to take $\Omega(\log^* n)$ rounds~\cite{linial92locality,naor91lower}.

\section{Preliminaries}\label{sec:prelim}

Let $G = (V,E)$ be a simple graph. An \emph{orientation} $\sigma$ of a graph $G$ assigns a direction $\sigma(\{u,v\}) \in \{u \rightarrow v, u \leftarrow v\}$ for each edge $\{u,v\} \in E$. For convenience, we write $(u,v) \in \sigma(E)$ to denote an edge $\{u,v\} \in E$ oriented $u \rightarrow v$ by $\sigma$. For all $v \in V$ we define $\indeg(v, \sigma) = |\{u : (u,v) \in \sigma(E) \}|$ as the number of incoming edges, $\outdeg(v, \sigma) = |\{u : (v,u) \in \sigma(E) \}|$ as the number of outgoing edges, and $\deg(v) = \indeg(v, \sigma) + \outdeg(v, \sigma)$ as the degree of $v$. Graph $G$ is $d$-regular if for all $v \in V$ we have $\deg(v) = d$.

A node $v$ with $\indeg(v, \sigma) = \deg(v)$ is called a \emph{sink}. We call an orientation $\sigma$ \emph{sinkless} if no node is a sink, that is, every node $v$ has $\outdeg(v, \sigma) > 0$.

\subsection{Colourings}

In the following, for any integer $k>0$ we write $[k] = \{0,1, \ldots, k-1\}$. A function $\varphi \colon V \to [\chi]$ is a \emph{proper node $\chi$-colouring} if for all $\{u,v\} \in E$ we have $\varphi(u) \neq \varphi(v)$. We say that $\psi \colon E \to [\chi]$ is a \emph{proper edge $\chi$-colouring} if for all $e,e' \in E$, where $e \neq e'$, it holds that $e \cap e' \neq \emptyset \Rightarrow \psi(e) \neq \psi(e')$. That is, any two adjacent edges have a different colour.

Given a properly edge $\chi$-coloured graph $G = (V, E, \psi)$, we call $\varphi \colon V \to [\chi]$ a \emph{sinkless colouring} of $G$ if for all edges $e = \{u,v\} \in E$ it holds that $\varphi(u) = \psi(e) \Rightarrow \varphi(v) \neq \psi(e)$. Put otherwise, $\varphi$ is a sinkless colouring if it does not contain a \emph{forbidden configuration}, where $\varphi(u) = \varphi(v) = \psi(e)$ for some edge $e = \{u,v\} \in E$. Note that a sinkless colouring is not necessarily a proper node colouring. The name ``sinkless colouring'' refers to its close relation to sinkless orientations (see \sectionref{ssec:sinkless-definitions}).

\subsection{Model of Computation}

In this work, we consider the $\local$ model of distributed computing~\cite{linial92locality,peleg00distributed}. In this framework, we have a simple connected undirected graph $G=(V,E)$ that serves both as a communication network and as the problem instance. Each node $v \in V$ is a computational unit, edges denote direct communication links between nodes, and all nodes in the system execute the same algorithm~$A$.

Initially each node~$v$ knows the total number of nodes $n$, the maximum degree of the graph~$\Delta$, and possibly a task-specific local input. Computation proceeds in synchronous rounds. In each round, every node performs the following three steps:
\begin{enumerate}[noitemsep]
  \item send a message to each neighbour,
  \item receive a message from each neighbour,
  \item perform local computation.
\end{enumerate}
After the final round each node announces its own local output, that is, its own part in the solution. We do not bound the local computation performed by nodes each round in any way or the size of the messages sent. In particular, nodes can send \emph{infinitely} long messages to their neighbours in a single round. The running time of an algorithm is defined to be the number of communication rounds needed for all nodes to announce their local output.

In the case of \emph{randomised} algorithms, we assume that each node can toss a countably infinite number of random coins, or equivalently, is provided with a real number $x(v)$ taken uniformly at random from the interval $[0,1]$. 
Note that $x(v)$ provides a globally unique identifier with probability 1 and it can be used to obtain a globally unique $O(\log n)$-bit identifier with high probability.

We emphasise that while some of our assumptions may not be realistic, they only make the lower-bound result stronger.

\subsection{Local Neighbourhoods}

We denote the radius-$t$ neighbourhood of a node~$u$ by
$
N^t(u) = \set{v \in V}{\dist(u,v) \le t}
$,
where $\dist(u,v)$ is the length of the shortest path between $u$ and $v$. Note that in $t$ rounds, any node $u$ can only gather information from $t$ hops away, and hence, has to decide its output based solely on $N^t(u)$. Thus, any distributed $t$-time algorithm can be considered as a function that maps the radius-$t$ neighbourhoods to output values.

It is often convenient to consider the \emph{edges} instead of the nodes as active entities, that is, every edge outputs e.g.\ its own orientation. Hence, we analogously define the radius-$t$ neighbourhood of an \emph{edge} $\{u,v\}$ as $N^t(\{u,v\}) = N^t(u) \cap N^t(v)$. Now for any algorithm that runs in time $t$, the output of an edge $\{u,v\} \in E$ can only depend on $N^t(\{u,v\})$.

\subsection{Distributed Sinkless Orientation and Sinkless Colouring}\label{ssec:sinkless-definitions}

In order to prove our lower bound, we consider the tasks of finding a sinkless orientation and a sinkless colouring in $d$-regular graphs with distributed algorithms. In particular, we will show that solving these problems is hard even in the case where we are given an edge $d$-colouring. In the following, let $G = (V, E, \psi)$ denote our input graph, where $\psi$ is a proper edge $d$-colouring; see \figureref{fig:problems} for illustrations.

\begin{figure}
\centering
\includegraphics[page=1]{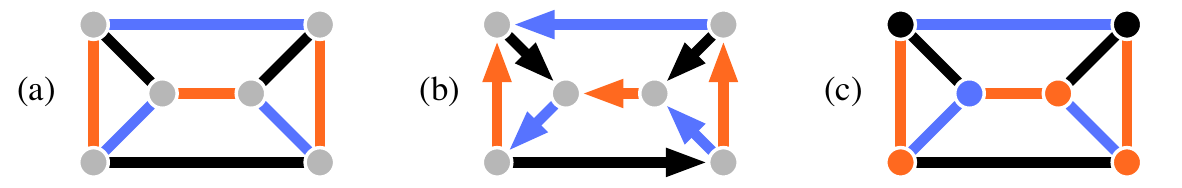}
\caption{(a)~A $3$-regular edge $3$-coloured graph. (b)~A sinkless orientation. (c)~A sinkless colouring. Note that the solutions (b) and (c) are closely related: if the colour of a node is $c$ in Figure~(c), then its incident edge of colour $c$ is one of its out-edges in Figure~(b).}\label{fig:problems}
\end{figure}

\begin{problem}[Sinkless colouring]
Given an edge $d$-coloured $d$-regular graph $G = (V, E, \psi)$, find a sinkless colouring $\varphi$. That is, compute a colouring $\varphi$ such that for no edge $e = \{u,v\} \in E$ we have $\varphi(u) = \varphi(v) = \psi(e)$. 
\end{problem}

In the sinkless colouring problem, each node $v \in V$ only outputs its own colour $\varphi(v)$ in the computed colouring. 

\begin{problem}[Sinkless orientation]
 Given an edge $d$-coloured $d$-regular graph $G = (V, E, \psi)$, find a sinkless orientation. That is, compute an orientation $\sigma$ such that $\outdeg(v, \sigma) > 0$ for all $v \in V$.
\end{problem}

Note that in the sinkless orientation problem, the output relates to \emph{edges} instead of nodes. Therefore, we require that for an edge $e = \{u,v\}$ both endpoints $u$ and $v$ agree on the orientation (i.e., either $u \rightarrow v$ or $u \leftarrow v$) and output the same value $\sigma(e)$ with probability 1. Put otherwise, $u$ and $v$ have to break symmetry in order not to have the nodes trying to orient the edge in a conflicting manner, e.g.\ outwards from themselves. In the case of randomised algorithms, the random input value $x(v) \in [0,1]$ breaks symmetry with probability 1; the event $x(v) = x(u)$ occurs with probability 0, so we simply ignore this case for the remainder of this paper.

Finally, note that the problems of sinkless colouring and sinkless orientation are closely related (see \figureref{fig:problems}). Given a sinkless colouring $\varphi$, node $u$ can orient the edge $\{u,v\}$ with colour $\varphi(u)$ towards $v$ and inform $v$ of this in one communication round; edges that are still unoriented can be oriented arbitrarily. This produces a sinkless orientation. On the other hand, given a sinkless orientation $\sigma$ we can compute a sinkless colouring $\varphi$ as follows: node $u$ outputs the smallest colour $\psi(e)$, where $e = \{u,v\}$ is an outgoing edge, that is, $\sigma(e) = u \rightarrow v$. Hence we have the following trivial observations:
\begin{enumerate}[noitemsep]
    \item\label{step:triva} If we can find a sinkless orientation in $t$ rounds, we can find a sinkless colouring in $t$ rounds.
    \item\label{step:trivb} If we can find a sinkless colouring in $t$ rounds, we can find a sinkless orientation in $t+1$ rounds.
\end{enumerate}
The mutual speedup lemma (see \sectionref{sec:speedup}) shows that we can save $1$ communication round in both steps \ref{step:triva} and \ref{step:trivb}, at least in high-girth graphs.

\subsection{Distributed \LLLemma}\label{sec:distlll}

Let $\mathcal{X}$ be the set of random variables and $\mathcal{E} = \{ E_1, \ldots, E_n \}$ be the set of events as in \theoremref{thm:lll} in \sectionref{ssec:dist-lll}. Denote by $\vbl(E_k) \subseteq \mathcal{X}$ the subset of variables that event $E_k \in \mathcal{E}$ depends on and the dependency graph by $G_\mathcal{E} = (\mathcal{E}, \mathcal{D})$, where $\mathcal{D} = \set{\set{E_i,E_j}}{\vbl(E_i) \cap \vbl(E_j) \neq \emptyset}$. 

\begin{problem}[Distributed \lllemma]
 Let the dependency graph $G_\mathcal{E} = (\mathcal{E}, \mathcal{D})$ be the communication graph, where each node $v$ corresponds to an event $E_v \in \mathcal{E}$ and knows the set $\vbl(E_v)$. The task is to have each node output an assignment $a_v$ of the variables $\vbl(E_v)$ such that
\begin{enumerate}[noitemsep]
 \item for any $\{E_u, E_v\} \in \mathcal{D}$ and $X \in \vbl(E_u) \cap \vbl(E_v)$ it holds that $a_u(X) = a_v(X)$, and
 \item the event $E_v$ does not occur under assignment $a_v$.
\end{enumerate}
\end{problem}

To make the lower-bound results as widely applicable as possible, we consider the \emph{explicit, finite version} of the LLL problem: the random variables are discrete variables with a finite range, and for each event $E_v$, node $v$ has access to an explicit specification of all combinations of the variables $\vbl(E_v)$ for which $E_v$ occurs. In particular, we do not need to assume that the events are black boxes.

\section{From LLL to Sinkless Orientation}\label{sec:reduction}

In this section, we reduce the sinkless orientation problem to the distributed \lllemma{}. More specifically, we show the following:

\begin{theorem}
  Let $f\colon \N \to \Real$ be such that $f(4) \leq 16$. Let $A$ be a Monte Carlo distributed algorithm for LLL such that $A$ finds an assignment avoiding all the bad events under the LLL criteria $pf(d) \leq 1$ in time $T$ for some $T\colon \N \to \N$. Then there is a Monte Carlo distributed algorithm~$B$ that finds a sinkless orientation in $3$-regular graphs of girth at least 5 in time $O(T)$.
  \label{thm:sotolll}
\end{theorem}

The case for 4-regular graphs will be almost immediate. The interesting case will be 3-regular graphs, for which we will show a reduction to the 4-regular case.

\subsection{Sinkless Orientation in 4-regular Graphs}

In the sinkless orientation problem, the output is an orientation for each edge $e \in E$ in a graph $G=(V,E)$---we will not need an edge colouring in this case. Thus, we associate each node $v \in V$ with $\deg(v)$ variables, one for each of its edges. Put otherwise, we have $\vbl(E_v) = \set{X_e}{v \in e}$ for each $v \in V$. That is, two events $E_u$ and $E_v$ share a variable if and only if $\set{u,v} \in E$. For each edge $e = \set{u,v} \in E$, the corresponding variable $X_e$ ranges over $\set{u \rightarrow v, u \leftarrow v}$. Since the task is to find a \emph{sinkless} orientation, the bad event $E_v$ occurs exactly when for all neighbours $u$ of $v$ the variable $X_{\set{v,u}}$ takes the value $u \rightarrow v$. Now our setting is as described in \sectionref{sec:distlll}: the dependency graph~$G_{\Ep} = (\mathcal{E}, \mathcal{D})$ is isomorphic to the sinkless orientation problem instance~$G=(V,E)$, which also acts as the communication graph.

Consider now finding a sinkless orientation in a 4-regular graph. If the mutually independent random variables $X_e$ are sampled uniformly at random, we have
$
  \Pr[E_v] = 1/2^4 = 1/16
$
for each $v \in V$. 
Let $p = 1/16$ and $d = 4$. Now $\Pr(E_v) \leq p$ and $E_v$ depends on $d$ other events for each $v \in V$, and the condition $pf(d) \leq 1$ holds, given $f(4) \leq 16$.

Let $A$ be the algorithm from the statement of \theoremref{thm:sotolll} and let $a_v$ be the assignment produced by $A$ for each $v \in V$. By assumption, we have $a_v(X_e) = a_u(X_e)$ for all $e = \set{u,v} \in E$ and none of the events $E_v$ occur. Define an orientation $\sigma$ of $G$ by setting $\sigma(e) = a_v(X_e)$, where $v \in e$, for each $e \in E$. By the definition of the events $E_v$, the orientation $\sigma$ is sinkless. We have shown that $A$ finds a sinkless orientation in 4-regular graphs in time $T$. Note that we did not make use of the edge colouring in any way; this will be crucial in the next step.

\subsection{Sinkless Orientation in 3-regular Graphs}

Unfortunately, the \lllemma{} is not directly applicable in the case of 3-regular sinkless orientation, since the probability of bad events will be $p = 1/2^3 = 1/8$ and thus $ep(d+1) > 1$. We will circumvent this by reducing the 3-regular case to the 4-regular case.

The key observation is that a pair of adjacent degree-3 nodes can act as a single virtual degree-4 node: if we contract two adjacent nodes into one by their common edge, the new node has 4 edges; see \figureref{fig:contract}. If the degree-4 node is not a sink in an orientation, we can direct the contracted edge in the original graph such that neither of the original nodes is a sink.

\begin{figure}
  \centering
  \includegraphics[page=2]{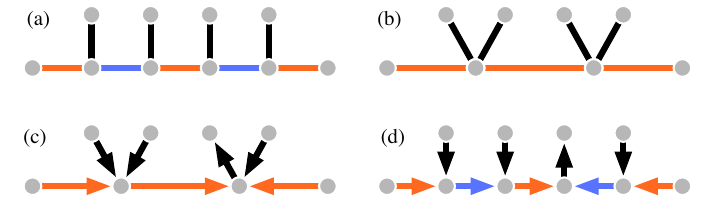}
  \caption{(a)~Part of a high-girth $3$-regular edge $3$-coloured graph. (b)~A $4$-regular graph obtained by contracting all blue edges---note that the graph is not properly edge coloured. (c)~With distributed LLL we can find a sinkless orientation in the $4$-regular graph. (d)~We can now orient each blue edge greedily to obtain a sinkless orientation of the original graph.}\label{fig:contract}
\end{figure}

The proper edge 3-colouring given as input can be used to form the node pairs. Since each node has exactly one incident edge of each colour, it follows that the edges of each colour class constitute a perfect matching. Thus, we can choose to contract all edges of colour 2 to obtain a 4-regular graph. Our algorithm for the 3-regular case works by first simulating the algorithm described in the previous section in the ``virtual'' 4-regular graph and then translating the obtained sinkless orientation to the actual 3-regular graph. We do not get a proper edge colouring for the 4-regular graph, but as stated previously, it is not necessary.

Let $G = (V,E,\psi)$ be an edge 3-coloured 3-regular graph with a girth at least 5. Denote by $G' = (V',E')$ the graph for which
\begin{align*}
 V' &= \set{\set{u,v} \in E}{\psi(\set{u,v}) = 2}, \\
 E' &= \set{\set{a,b}}{a \neq b \text{ and } \set{u,v} \in E \text{ for some } u \in a, v \in b, \text{ and } \psi(\set{u,v}) < 2}.
\end{align*}
The following lemma shows that graph~$G'$ is an eligible problem instance for algorithm~$A$ of the previous section.

\begin{lemma}\label{lem:4reg}
  The graph $G'=(V',E')$ is simple and 4-regular. Furthermore, for each edge $\set{a,b} \in E'$ there is only one edge $\set{u,v} \in E$ with $u \in a, v \in b$ and $\psi(\set{u,v}) < 2$, and vice versa.
\end{lemma}

\begin{proof}
  The fact that $G'$ is simple follows directly from the definition. Suppose then that $a = \set{u,v} \in V'$. As $G$ is a 3-regular edge 3-coloured graph, both $u$ and $v$ have exactly two edges of colour 0 or 1. Let us denote those edges by $e_{u,0} = \set{u,u'}$, $e_{u,1} = \set{u,u''}$, $e_{v,0} = \set{v,v'}$, and $e_{v,1} = \set{v,v''}$. It follows that $\deg_{G'}(a) \leq 4$. Let $b \in V'$, $b' \in V'$, $c \in V'$, and $c' \in V'$ be such that $u' \in b$, $u'' \in b'$, $v' \in c$, and $v'' \in c'$. Now we have $\set{a,b} \in E'$, $\set{a,b'} \in E'$, $\set{a,c} \in E'$, and $\set{a,c'} \in E'$. If $|\set{b,b',c,c'}| < 4$, that is, some of the nodes are equal, it follows that we can find a cycle of length at most 4 in $G$, a contradiction. Hence we have $\deg_{G'}(a) \geq 4$, and consequently, $G'$ is 4-regular. The second claim follows from the definition of $E'$ and the 4-regularity of $G'$.
\end{proof}

\lemmaref{lem:4reg} also implies that we can define a bijective mapping $\ell\colon E' \to \set{e \in E}{\psi(e) < 2}$ by setting $\ell(\set{a,b}) = \set{u,v}$, where $u \in a$ and $v \in b$.

In what follows, we construct a Monte Carlo distributed algorithm $B$ that runs in graph~$G$ and simulates algorithm~$A$ in graph~$G'$. First, we use the random number $x(v)$ of each node $v$ to construct two uniformly chosen random numbers $y(v) \in [0,1]$ and $z(v) \in [0,1]$---for instance, use bits of odd indices for $y(v)$ and bits in even indices for $z(v)$. Then, for each edge $e = \set{u,v} \in E$ with $\psi(e) = 2$ we select a \emph{leader} $L(e) \in e$ by setting $L(e) = \arg\max_{v \in e} y(v)$. We call the other node $w \in e \setminus \set{L(e)}$ a \emph{relay} and write $R(e) = w$. The leader will take care of running the simulation for node $\set{u,v}$ of $G'$, while the task of the relay is just to forward messages. The other random number $z(L(e))$ of each leader $L(e)$ will be given to algorithm $A$ in the simulation.

\begin{lemma}\label{lem:path3}
  If nodes $a$ and $b$ are adjacent in graph~$G'$, there is a path of length at most 3 between nodes $L(a)$ and $L(b)$ in graph~$G$.
\end{lemma}

\begin{proof}
  Suppose that $\set{a,b} \in E'$. Then there is $\set{u,v} \in E$ such that $u \in a$ and $v \in b$. There are also $u',v' \in V$ such that $a = \set{u,u'} \in E$ and $b = \set{v,v'} \in E$. Now $(u',\set{u',u},u,\set{u,v},v,\set{v,v'},v')$ is a path of length 3 in $G$ and contains both $L(a) \in \set{u,u'}$ and $L(b) \in\set{v,v'}$. The claim follows.
\end{proof}

\lemmaref{lem:path3} implies that to simulate one communication round of algorithm~$A$ in $G'$, algorithm~$B$ needs at most three communication rounds in $G$. Since we have $|V'| = n/2$, the entire simulation of $A$ takes $3T(n/2) \in O(T(n))$ rounds, where $T(n)$ is the running time of $A$. At the end of the execution of $A$, we have a sinkless orientation~$\sigma'$ of~$G'$. Algorithm~$B$ then produces an orientation~$\sigma$ of $G$ as follows: For each $e \in E$ with $\psi(e) < 2$ set $\sigma(e) = \sigma'(\ell^{-1}(e))$. For each $e \in E$ with $\psi(e) = 2$ set $\sigma(e) = L(e) \rightarrow R(e)$ if an edge of colour 0 or 1 was oriented away from $R(e)$ in the previous step, otherwise set $\sigma(e) = L(e) \leftarrow R(e)$.

\begin{lemma}
  The orientation~$\sigma$ produced by algorithm~$B$ is a sinkless orientation of graph~$G$.
\end{lemma}

\begin{proof}
  Consider a node $u \in V$. Let $v \in V$ be such that $a = \set{u,v} \in V'$. Suppose first that $L(\set{u,v}) = u$. If $\sigma(\set{u,v}) = u \rightarrow v$, then $u$ is not a sink. Otherwise, by the definition of orientation~$\sigma$, the edges of colour 0 and 1 incident to $v$ are oriented towards $v$. Denote these edges by $e_0$ and $e_1$, respectively. Now the edges $\ell^{-1}(e_0) \in E'$ and $\ell^{-1}(e_1) \in E'$ are oriented towards node~$a$ in the orientation~$\sigma'$, where $\ell$ is the bijective mapping given by \lemmaref{lem:4reg}. Since $\sigma'$ is assumed to be sinkless, we have $\sigma'(\set{a,b}) = a \rightarrow b$ for some $\set{a,b} \in E'$. Now $\ell(\set{a,b}) = \set{u,u'}$ for some $u' \in b$. By the definition of $\sigma$, we have $\sigma(\set{u,u'}) = u \rightarrow u'$, and hence $u$ is not a sink.

  Suppose then that $L(\set{u,v}) = v$. If $\sigma(\set{u,v}) = u \rightarrow v$, then again $u$ is not a sink. Otherwise it follows from the definition of $\sigma$ that there is an edge of colour 0 or 1 that is oriented away from $u$, and thus $u$ is not a sink. Since $u \in V$ was arbitrary, we have shown that $\sigma$ is a sinkless orientation.
\end{proof}

The output of $v$ according to algorithm~$B$ consists of the value $\sigma(\set{v,u})$ for each neighbour $u$ of~$v$. Electing the leaders as well as constructing the orientation $\sigma$ from $\sigma'$ can clearly be done in a constant number of rounds. Hence, for some constant~$C$, the total running time of $B$ is $O(T(n))+C \subseteq O(T(n))$. This completes the proof of \theoremref{thm:sotolll}.

\section{The Mutual Speedup Lemma \label{sec:speedup}}

In this section we show that if we can find a sinkless colouring in $t$ rounds with failure probability $p$, then it is possible to find a sinkless orientation in $t$ rounds with failure probability roughly $p^{1/3}$ assuming the graph has high girth. Furthermore, if we can find a sinkless orientation in $t$ rounds with failure probability $q$, then it is possible to find a sinkless colouring in $t-1$ rounds with failure probability roughly $q^{1/4}$.

We will assume throughout this section that the input graph is a 3-regular edge 3-coloured graph with a girth larger than $2t+1$. Thus, for each edge $e$, its radius-$(t+1)$ neighbourhood is a tree. In addition, we say that the radius-$t$ neighbourhood $N^t(u)$ of node $u$ is \emph{fixed} when we fix the random input values (i.e.,\ coin flips) for nodes in $N^t(u)$. Similarly, the radius-$t$ neighbourhood of an edge $e = \{u,v\}$ is fixed if all the random values in $N^t(e)$ are fixed. We denote the probability of an event $\Ep$ conditioned on a fixed radius-$t$ neighbourhood of $u$ by $\Pr[\Ep \mid N^t(u)]$, and respectively, for an edge $e$ by $\Pr[\Ep \mid N^t(e)]$.

\subsection{From Sinkless Colouring to Sinkless Orientation\label{ssec:sc-to-so}}

In this section, we assume we are given a randomised sinkless \emph{colouring} algorithm $B$ that runs in $t$ rounds. We use this algorithm to construct a randomised sinkless \emph{orientation} algorithm $B'$ that also runs in $t$ rounds. For shorthand, we write $B(u)$ for the colour that $u$ outputs according to $B$ and $B'(e)$ for the orientation $B'$ outputs for edge $e$.

Consider any node $u \in V$. Algorithm $B'$ consists of three steps. First, node $u$ gathers its radius-$t$ neighbourhood $N^t(u)$ in $t$ communication rounds. Note that $N^t(u)$ contains both the topology (which is locally a tree) \emph{and} the random input values $x(v)$ for all $v \in N^t(u)$. Second, node $u$ computes the set $C(u)$ of \emph{candidate colours} defined as 
\[
 C(u) = \{ \psi(e) : \Pr[ B(u) = \psi(e) \mid N^{t}(e) ] \ge K \textrm{ and } e = \{u,v\}\},
\]
where $K$ is a parameter we will fix later (see \figureref{fig:c2o} for an illustration).
Then, for all of its incident edges $e = \{u,v\}$, node $u$ calculates the probability of node $v$ outputting the colour $\psi(e)$ when executing algorithm $B$ given the radius-$t$ neighbourhood $N^t(e) = N^t(u) \cap N^t(v)$ of edge~$e$.
Thereby, $u$ can determine whether $\psi(e) \in C(v)$.

\begin{figure}
\centering
\includegraphics[page=3]{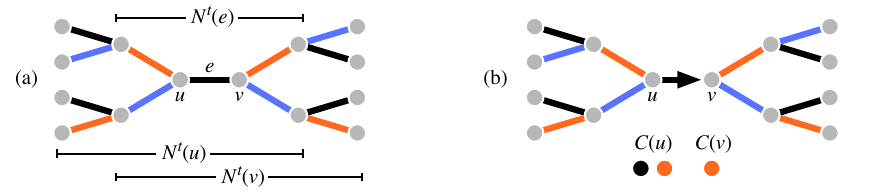}
\caption{From sinkless colouring to sinkless orientation. Here we are given a sinkless colouring algorithm~$B$ with a running time of $t = 2$, and our goal is to construct a sinkless orientation algorithm $B'$ with the same running time. (a)~In algorithm $B$, the colour of node $u$ is determined by the random bits in $N^t(u)$. However, when algorithm $B'$ chooses the orientation of the black edge $e$ it will only look at the random bits in $N^t(e) \subsetneq N^t(u)$. We say that black is a \emph{candidate colour} of $u$ if, based on the information in $N^t(e)$ alone, the probability of node $u$ outputting black in algorithm $B$ is at least $K$. (b)~If black is one of the candidate colours of $u$, and it is not one of the candidate colours of $v$, algorithm $B'$ will orient the edge $u \rightarrow v$.}\label{fig:c2o}
\end{figure}

Finally, we decide the orientation of each edge $e = \{u,v\}$ as follows. In the case $\psi(e) \in C(u) \cap C(v)$ or $\psi(e) \notin C(u) \cup C(v)$, choose the orientation $B'(e)$ of edge $e$ arbitrarily and break any ties using the random coin flips of $u$ and $v$. Otherwise, edge $e$ is oriented according to the following rule:
\[
 B'(e) = \begin{cases}
          u \rightarrow v & \textrm{if } \psi(e) \in C(u) \textrm{ and } \psi(e) \notin C(v), \\
          u \leftarrow v  & \textrm{if }  \psi(e) \notin C(u) \textrm{ and } \psi(e) \in C(v).
         \end{cases}
\]

We will now analyse algorithm $B'$. For each colour $c$, let $e = \{u,v\}$ be an edge incident to $u$ such that $\psi(e) = c$ and define
\[
 A_c(u) = N^{t-1}(v) \setminus N^{t-1}(u).
\]
\begin{definition}[Lucky random bits]\label{def:lucky}
Given a fixed radius-$(t-1)$ neighbourhood $N^{t-1}(u)$ of $u$, we say that the random coin flips in $A_c(u)$ are \emph{lucky} if
\[
 \Pr[B(u) = c \mid N^{t-1}(u) \cup A_c(u)] \ge K
\]
holds, and otherwise that the coin flips are \emph{unlucky}. 
\end{definition}
Observe that $c \in C(u)$ if and only if the random coin flips in $A_c(u)$ are lucky, since $N^{t-1}(u) \cup A_c(u) = N^t(e)$.
Let $\Ep_c$ be the event that the random coin flips in $A_c(u)$ are unlucky, that is, the event that $\Pr[B(u) = c \mid N^{t-1}(u) \cup A_c(u)] < K$ holds. 

\begin{lemma}
 Given any fixed neighbourhood $N^{t-1}(u)$ of node $u$, the set $C(u)$ is empty with probability at most $3K$.
\label{lemma: coloursetnotempty}
\end{lemma}
\begin{proof}
Let $\Ep = \bigcap \Ep_c$ be the event that the random values in each $A_c(u)$ are unlucky given $N^{t-1}(u)$. This is the case if and only if $C(u) = \emptyset$, which implies that
\[
 \Pr[C(u) = \emptyset \mid N^{t-1}(u)] = \Pr[\Ep],
\]
where the right-hand side can be written as
\[
 \Pr[\Ep] = \sum_{c} \Pr[\Ep \textrm{ and } B(u) = c]. 
\]
Observe that since $\Ep \subseteq \Ep_c$ for any colour $c$, we have that
\begin{align*}
 \Pr[\Ep \textrm{ and } B(u) = c] &= \Pr[\Ep \textrm{ and } B(u) = c \mid \Ep_c] \cdot \Pr[\Ep_c] \\
                         &\le \Pr[B(u) = c \mid \Ep_c] \cdot \Pr[\Ep_c] \\
						 &\le \Pr[B(u) = c \mid \Ep_c].
\end{align*}
Since by definition the coin flips in $A_c(u)$ are unlucky in the event $\Ep_c$, we get that $\Pr[B(u) = c \mid \Ep_c] < K$. Thus combining the above, we have that
\[  
\Pr[C(u) = \emptyset \mid N^{t-1}(u)] = \sum_{c} \Pr[\Ep \textrm{ and } B(u) = c] < \sum_{c} K.
\]
Since we have three colours, the claim follows.
\end{proof}

\begin{definition}[Nice edge neighbourhoods]\label{def:niceness}
For an edge $e = \{u,v\}$, we call its fixed neighbourhood $N^{t}(e)$ \emph{nice} if 
\[
 \Pr[B(u) = \psi(e) = B(v) \mid N^t(e)] < K^2.
\]
That is, after fixing the random coin flips in $N^t(e)$, the algorithm outputs $\psi(e)$ with probability less than $K^2$. Otherwise, we call $N^t(e)$ a \emph{bad} neighbourhood. 
\end{definition}

\begin{lemma}
Let $e = \{u,v\}$ be an edge with no cycles in its radius-$(t+1)$ neighbourhood. If the fixed neighbourhood $N^t(e)$ is nice, then $\psi(e) \notin C(u) \cap C(v)$.
 \label{lemma: candc case}
\end{lemma}
\begin{proof}
Let $N^t(e)$ be fixed and nice. Assume for contradiction that $\psi(e) \in C(u) \cap C(v)$. By definition of the candidate colour set, for both $w \in e$ we have $\psi(e) \in C(w)$ if
\[
 \Pr[B(w) = \psi(e) \mid N^t(e)] \ge K.
\]
As the output $B(u)$ of node $u$ is determined by the coin flips in $N^t(u)$ and the coin flips in $N^t(e) = N^t(u) \cap N^t(v)$ are fixed, we now have that $B(u)$ only depends on the coin flips in $N^t(u) \setminus N^t(v)$. Similarly, the output $B(v)$ of $v$ only depends on the coin flips in $N^t(v) \setminus N^t(u)$. Therefore, the events $B(u) = \psi(e)$ and $B(v) = \psi(e)$ are independent, as $N^{t+1}(e)$ contains no cycles, and we get
\[
 \Pr[B(u) = \psi(e) = B(v) \mid N^t(e)] \ge K^2,
\]
which contradicts the assumption that $N^t(e)$ was nice.
\end{proof}

Now it is easy to check that if a node $u$ has at least  one candidate colour and all its incident edges have nice neighbourhoods, then $u$ will not be a sink according to $B'$.

\begin{lemma}\label{lemma:not-a-sink}
 Suppose $N^{t}(u)$ is fixed and the neighbourhoods $N^t(e)$ are nice for all edges $e = \{u,w\}$ incident to $u$. If $C(u) \neq \emptyset$, then $B'(e') = u \rightarrow v$ for some $e' = \{u,v\}$.
\end{lemma}
\begin{proof}
 Since $C(u) \neq \emptyset$, there is some $\psi(e) \in C(u)$. Moreover as $N^t(e)$ is nice, \lemmaref{lemma: candc case} implies that $\psi(e) \notin C(u) \cap C(v)$, and thus, $\psi(e) \notin C(v)$. By definition of $B'$, we have $B'(e) = u \rightarrow v$.
\end{proof}

Now we have all the pieces to show the first part of the mutual speedup lemma.

\begin{lemma}
Suppose $B$ is a sinkless colouring algorithm that runs in $t$ rounds such that for any edge $e = \{u,v\}$ the probability of outputting a forbidden configuration $B(u) = \psi(e) = B(v)$ is at most $p$. Then there exists a sinkless orientation algorithm $B'$ that runs in $t$ rounds such that for any node $u$ the probability of being a sink is at most $6p^{1/3}$.
  \label{lemma: SCtoSO}
\end{lemma}
\begin{proof}
Let $B'$ be as given earlier and consider a node $u$. 
By \lemmaref{lemma:not-a-sink}, algorithm $B'$ can produce a sink at node $u$ only if $C(u) = \emptyset$ or one of the edges incident to $u$ has a bad (i.e., not nice) neighbourhood. Let $S = \max_e \Pr[N^t(e) \textrm{ is bad}]$ be the maximum probability that some edge has a bad neighbourhood; the probability of having a bad neighbourhood need not be the same for edges of different colours. By the union bound, the probability that $N^t(e)$ is bad for \emph{some} edge $e = \{u,v\}$ is at most $3S$. By \lemmaref{lemma: coloursetnotempty}, the probability that $C(u) = \emptyset$ is at most $3K$. Thus, applying the union bound once again, we get that 
\begin{equation*}\label{eq:sink-probability}
 \Pr[\textrm{node } u \textrm{ is a sink}] \le \sum_{e = \{u,v\}} \Pr[N^t(e) \textrm{ is bad}] + \Pr[C(u) = \emptyset] \le 3S + 3K.
\end{equation*}
Now let us consider the probability that an edge $e = \{u,v\}$ has a forbidden configuration, where $e$ is an edge that attains $\Pr[N^t(e) \textrm{ is bad}] = S$. Recall that the probability of $B(u) = \psi(e) = B(v)$ is at most $p$, and thus,
\begin{align*}
 p &\ge \Pr[B(u) = \psi(e) = B(v)] \\
   &\ge \Pr[N^t(e) \textrm{ is bad}] \cdot \Pr[B(u) = \psi(e) = B(v) \mid N^t(e) \textrm{ is bad}] \\
   &\ge SK^2
\end{align*}
by \defref{def:niceness}. By setting $K = p^{1/3}$ we get that 
\[
 p \ge SK^2 = Sp^{2/3} \iff p^{1/3} \ge S
\]
and we have $3S+3K \le 6p^{1/3}$ which proves our claim.
\end{proof}

\subsection{From Sinkless Orientation Back to Sinkless Colouring}

We now show how to construct a randomised sinkless \emph{colouring} algorithm $B''$ that runs in time $t-1$ given a sinkless \emph{orientation} algorithm $B'$ that runs in time $t$. The approach is analogous to the one in the previous section. The high level idea is that any node $u$ first checks which of its incident edges are likely to be pointed outwards by $B'$, and then it can choose the colour of one of these edges to output a sinkless colouring with a large probability.

Unlike before, each node will gather only its radius-$(t-1)$ neighbourhood in $t-1$ rounds. Again, let $L$ be a threshold we fix later. Define the candidate colour set $C'(u)$ as 
\[
 C'(u) = \{ \psi(e) : \Pr[B'(e) = u \leftarrow v \mid N^{t-1}(u)] \le L \},
\]
that is, the set of colours which are pointed towards $u$ with probability at most $L$.
The node $u$ will then output the smallest candidate colour or an arbitrarily chosen colour if there are no candidates, or formally,
\[
  B''(u) = \begin{cases}
            \min C'(u) & \textrm{if } C'(u) \neq \emptyset, \\
            0          & \textrm{otherwise.}
           \end{cases}
\]
Our goal now is to show that this produces a sinkless colouring with a large probability. To do this, we show that the probabilities of the following two events are large: (1) the candidate set being non-empty and (2) $\psi(e) \notin C(u) \cap C(v)$ for any edge $e = \{u,v\}$.

Analogously to \sectionref{ssec:sc-to-so}, we define the notions of lucky/unlucky bits and nice/bad neighbourhoods.

\begin{definition}[Lucky random bits]\label{def:orient-luckiness}
For any $e = \{u,v\}$, let $A_u(e) = N^{t-1}(u) \setminus N^{t-1}(v)$. We say that the random coin flips in $A_u(e)$ are \emph{lucky} if
\[
 \Pr[ B'(e) = u \leftarrow v \mid N^{t-1}(e) \cup A_u(e)] \le L.
\]
Otherwise, the coin flips in $A_u(e)$ are \emph{unlucky}.
\end{definition}

\begin{lemma}
\label{lemma: SOtoSCintersection}
Given any fixed neighbourhood $N^{t-1}(e)$ of edge $e$, we have
\[
\Pr[ \psi(e) \in C'(u) \cap C'(v) \mid N^{t-1}(e)] \le 2L.
\]
\end{lemma}
\begin{proof}
Fix the random coin flips in $N^{t-1}(e)$. Let $\Ep_u$ be the event that the coin flips in $A_u(e)$ are lucky and let $\Ep = \Ep_u \cap \Ep_v$ be the event that coin flips in both $A_u(e)$ and $A_v(e)$ are lucky. Observe that $\psi(e) \in C'(u)$ if and only if the coin flips in $A_u(e)$ are lucky. Therefore,
\begin{align*}
 \Pr[\Ep] &= \Pr[\psi(e) \in C'(u) \cap C'(v) \mid N^{t-1}(e)] \\
 &= \Pr[\Ep \textrm{ and } B'(e) = u \rightarrow  v] + \Pr[\Ep \textrm{ and } B'(e) = u \leftarrow v].
\end{align*}
Since $\Ep \subseteq \Ep_u$, it follows that 
\begin{align*}
 \Pr[\Ep \textrm{ and } B'(e) = u \leftarrow v] &= \Pr[\Ep \textrm{ and } B'(e) = u \leftarrow v \mid \Ep_u] \cdot \Pr[\Ep_u] \\
 &\le \Pr[B'(e) = u \leftarrow v \mid \Ep_u] \le L 
\end{align*}
by \defref{def:orient-luckiness} as the coin flips in $A_u(e)$ are lucky in the event $\Ep_u$. Symmetrically, we also get the bound $\Pr[\Ep \textrm{ and } B'(e) = u \rightarrow v] \le L$. Combining the above observations we get that
\[
 \Pr[\Ep] = \Pr[\psi(e) \in C'(u) \cap C'(v) \mid N^{t-1}(e)] \le 2L. \qedhere
\]
\end{proof}

\begin{definition}[Nice node neighbourhoods]\label{def:node-niceness}
 Let $N^{t-1}(u)$ be fixed. We say that the neighbourhood $N^{t-1}(u)$ is \emph{nice} if the probability that $u$ is a sink when executing $B'$ is at most $L^3$, that is, if
\[
 \Pr[B'(e) = u \leftarrow v \textrm{ for all } e = \{u,v\} \mid N^{t-1}(u)] \le L^3
\]
holds. Otherwise, we call $N^{t-1}(u)$ a \emph{bad} neighbourhood.
\end{definition}

\begin{lemma}
Assume that the fixed neighbourhood $N^{t-1}(u)$ is nice. Then $C'(u) \neq \emptyset$.
 \label{lemma: Cprimenotempty}
\end{lemma}
\begin{proof}
Fix the coin flips in $N^{t-1}(u)$ and assume $N^{t-1}(u)$ is nice. For the sake of contradiction, suppose $C'(u) = \emptyset$. Now by definition of $C'(u)$ we have
\[
 \Pr[B'(e) = u \leftarrow v \mid N^{t-1}(u)] > L
\]
for each edge $e = \{u,v\}$. Since the coin flips in $N^{t-1}(u)$ are fixed, the output $B'(e)$ only depends on the coin flips in $N^{t-1}(v) \setminus N^{t-1}(u)$. Since the girth is larger than $2t$, for each $e = \{u,v\}$ and $e' = \{u,v'\}$, where $v \neq v'$, the coin flips in $N^{t-1}(v) \setminus N^{t-1}(u)$ and $N^{t-1}(v') \setminus N^{t-1}(u)$ are independent. Therefore, the events $B'(e) = u \leftarrow v$ and $B'(e') = u \leftarrow v'$ are independent. This implies that 
\[
 \Pr[C'(u) = \emptyset \mid N^{t-1}(u)] = \prod_{e = \{u,v\}} \Pr[B'(e) = u \leftarrow v \mid N^{t-1}(u)] > L^3,
\]
contradicting the assumption that $N^{t-1}(u)$ is nice.
\end{proof}

\begin{lemma}
Suppose $B'$ is a sinkless \emph{orientation} algorithm that runs in time $t$ such that the probability that any node $u$ is a sink is at most $\ell$. Then there exists a sinkless \emph{colouring} algorithm $B''$ that runs in time $t-1$ such that the probability for any edge $e = \{u,v\}$ having a forbidden configuration $B''(u) = \psi(e) = B''(v)$ is less than $4\ell^{1/4}$.
  \label{lemma: SOtoSC}
\end{lemma}
\begin{proof}
Let $B''$ as defined earlier and consider an edge $e = \{u,v\}$. If algorithm $B''$ outputs a forbidden configuration $B''(u) = \psi(e) = B''(v)$, then either $C'(u) \cup C'(v) = \emptyset$ or $\psi(e) \in C'(u) \cap C'(v)$ holds. We will now bound the probability of both events.

Observe that before fixing any random bits, the probability of having a bad radius-$(t-1)$ neighbourhood is the same for all nodes, as all radius-$(t-1)$ node neighbourhoods are identical. Let $S = \Pr[N^{t-1}(u) \textrm{ is bad}]$ be this probability. By union bound and \lemmaref{lemma: Cprimenotempty} we get that
\begin{align*}
 \Pr[C'(u) \cup C'(v) = \emptyset] &\le \Pr[C'(u) = \emptyset] + \Pr[C'(v) = \emptyset] \\
                                   &\le \Pr[N^{t-1}(u) \textrm{ is bad}] + \Pr[N^{t-1}(v) \textrm{ is bad}] \\ 
&\le 2S.
\end{align*}
From \lemmaref{lemma: SOtoSCintersection} we get that 
\[
 \Pr[\psi(e) \in C'(u) \cap C'(v)] \le 2L.
\]
Using the union bound and the above, we get that the probability of a forbidden configuration is 
\[
 \Pr[B''(u) = \psi(e) = B''(v)] \le 2S+2L.
\]

To prove the claim, observe that from \defref{def:node-niceness} and the assumption that $B'$ produces a sink at $u$ with probability at most $\ell$, it follows that 
\[
 \ell \ge \Pr[u \textrm{ is a sink}] \ge \Pr[u \textrm{ is a sink} \mid N^{t-1}(u) \textrm{ is bad}] \cdot \Pr[N^{t-1}(u) \textrm{ is bad}] > SL^3.
\]
Therefore, $\ell > SL^3$. By setting $L = \ell^{1/4}$ we get that $S < \ell^{1/4}$ implying $2S + 2L < 4\ell^{1/4}$. 
\end{proof}

\subsection{The Speedup Lemma}

The following is an immediate consequence of \lemmaref{lemma: SCtoSO} and \lemmaref{lemma: SOtoSC}.

\begin{lemma}\label{lemma:speedup}
 Suppose $B$ is a sinkless colouring algorithm that runs in time $t$ such that for any edge $e$ the probability that $B$ produces a forbidden configuration at $e$ is at most $p$. Then there is a sinkless colouring algorithm $B''$ that runs in $t-1$ rounds such that it produces a forbidden configuration at any edge with probability less than $4\cdot6^{1/4}\cdot p^{1/12}$.
\end{lemma}

\section{Lower Bounds}

\begin{lemma}
\label{lemma:high-girth-graphs}
 Fix $d \ge 3$. There exists an infinite family of $d$-regular graphs $\mathcal{G}$ such that the edges of every $G \in \mathcal{G}$ can be coloured with $d$ colours and the girth of $G$ is $\Omega(\log n)$.
\end{lemma}
\begin{proof}
 Let $\mathcal{G}'$ be an infinite family of $d$-regular graphs with girth $\Theta(\log n)$; see e.g.\ \cite[Ch. 3]{bollobas78extremal} how to obtain one. For any $G' \in \mathcal{G}'$ consider its bipartite double cover $G$ which is also $d$-regular and has girth of $\Theta(\log n)$. By König's line colouring theorem the edges of $G$ can be coloured with $d$ colours~\cite[Ch. 5.3]{diestel10graph}.
\end{proof}

\begin{theorem}\label{thm:colouring-lb}
 There does not exist a Monte Carlo distributed algorithm solving the sinkless colouring problem in $d$-regular graphs with high probability in $o(\log \log n)$ rounds.
\end{theorem}
\begin{proof}
For the sake of contradiction, suppose $A$ is an algorithm that solves sinkless colouring in $f_c(n) \in o(\log \log n)$ rounds with probability at least $1 - 1/n^c$ for an arbitrarily large constant $c$. Now fix a sufficiently large 3-regular graph $G$ of $n$ nodes given by \lemmaref{lemma:high-girth-graphs}. 

Let $t = f_c(n)$ and for $i \in \{0, \ldots, t\}$ let $A_i$ be the algorithm attained after $i$ iterated applications of \lemmaref{lemma:speedup}. Let $p_i$ be the probability that $A_i$ produces a forbidden configuration at any given edge $e$. By assumption $p_0 \le 1/n^c$ and from \lemmaref{lemma:speedup} it follows that $p_{i+1} \le zp_i^{1/12}$, where $z = 4 \cdot 6^{1/4}$. In particular, the probability that algorithm $A_t$ running in 0 rounds produces a forbidden configuration at edge $e$ is
\[
 p_t \le z^{s} p_0^{1/12^{t}} \le z^{s} q(n,c), 
\]
where
\[
 s = \sum^{t}_{i=0} 1/12^{i} < 2 \quad \textrm{ and } \quad q(n,c) = n^{-c/(12^t)}.
\]
By applying the union bound we get that for any node $u$ executing $A_t$, the probability that one of its incident edges has a forbidden configuration is at most
\[
 \sum_{e : u \in e} z^2 q(n,c) < 3 z^2 q(n,c).
\]
Since $f_c(n) \in o(\log\log n)$, picking a sufficiently large $n$ yields $t = f_c(n) \le (\log\log n)/4$ and $1/12^{t} \ge 1/\log n$. It follows that $q(n,c) \le 1/2^{c}$ and by setting $c = \log(30z^2) \in O(1)$ we obtain
\[
\Pr[\textrm{node } u \textrm{ is incident to a forbidden configuration}] \le 3z^2 \cdot 1/2^{c} \le 1/10.
\]

Finally, observe that in 0 rounds, all nodes choose their output independently of each other (and using the same algorithm). Since each node needs to output a colour, at least one colour $c$ out of the three colours is picked with probability at least $1/3$. Now the probability that node $u$ has an edge of colour $c$ with a forbidden configuration is at least $1/3^2 > 1/10$, which is a contradiction.
\end{proof}

As observed in \sectionref{ssec:sinkless-definitions} we can obtain a sinkless colouring from a sinkless orientation without communication. This implies the following result.

\begin{corollary}
 There does not exist any Monte Carlo distributed algorithm solving the sinkless orientation problem in $d$-regular graphs with high probability in $o(\log \log n)$ rounds.
\end{corollary}

Together with \theoremref{thm:sotolll} we get our main result. Note that in the following theorem we can plug in, for example, either of the commonly used LLL criteria: $ep(d+1) \le 1$ or $4pd \leq 1$.

\begin{corollary}
  Let $f\colon \N \to \Real$ be such that $f(4) \leq 16$. Let $A$ be a Monte Carlo distributed algorithm for LLL that finds an assignment avoiding all the bad events under the LLL criteria $pf(d) \leq 1$ with high probability. Then the running time of $A$ is $\Omega(\log \log n)$ rounds.
\end{corollary}

Since any proper $d$-colouring of the nodes is also a sinkless colouring, we get the following lower bound as a by-product. Our lower bound can be contrasted with Linial's classical result~\cite{linial92locality}: Linial shows that any algorithm colouring a $d$-regular tree of radius $r$ in $2r/3$ rounds needs $\Omega(\sqrt{d})$ colours, and this result can be strengthened to $\Omega(d/\log d)$ colours using the graph constructions of Alon~\cite{alon10constant}. However, Linial's technique does not seem to imply any nontrivial lower bounds for the case of $d$ colours.

\begin{corollary}
 There does not exist any Monte Carlo distributed algorithm that finds a $d$-colouring in $d$-regular, bipartite, $\Omega(\log n)$-girth graphs with high probability in $o(\log \log n)$ rounds.
\end{corollary}

\section*{Acknowledgements}

The problem of sinkless orientations in the context of distributed computing was originally introduced by Mika Göös. We have discussed this problem and its variants with numerous people---including, at least, Laurent Feuilloley, Pierre Fraigniaud, Teemu Hankala, Joel Kaasinen, Petteri Kaski, Janne H.\ Korhonen, Juhana Laurinharju, Christoph Lenzen, Joseph S.\ B.\ Mitchell, Pekka Orponen, Thomas Sauerwald, Stefan Schmid, Przemysław Uznański, and Uri Zwick---many thanks to all of you!

\def\UrlFont{\sf\footnotesize}
\bibliographystyle{plainnat}
\bibliography{lll-lb}

\end{document}